\documentclass[11pt]{article}
\usepackage{amsmath,amsthm,amsfonts}
\usepackage{enumerate}
\usepackage{multirow}
\usepackage{bigstrut}
\usepackage{xspace}

\newtheorem{theorem}{Theorem}[section]
\newtheorem{lemma}[theorem]{Lemma}
\newtheorem{corollary}[theorem]{Corollary}
\newtheorem{claim}[theorem]{Claim}
\newtheorem{fact}[theorem]{Fact}
\newtheorem*{infthm}{Informal Statement}

\theoremstyle{definition}
\newtheorem{definition}[theorem]{Definition}
\newtheorem*{remark}{Remark}
\newtheorem{thm}{Theorem}[section]
\newtheorem{prop}[thm]{Proposition}
\newtheorem{defn}[thm]{Definition}

\theoremstyle{plain}

\newcommand{\ch}{c'}
\newcommand{\nh}{m}

\newcommand{\bit}{\{0,1\}}
\newcommand{\bits}[1]{\bit^{#1}}
\newcommand{\pmbit}{\{-1,1\}}
\newcommand{\pmbits}[1]{\pmbit^{#1}}
\newcommand{\iplus}{^{i+}}
\newcommand{\iminus}{^{i-}}

\newcommand{\cP}{\mathcal{P}}
\newcommand{\calP}{\mathcal{P}}

\newcommand{\CC}[2]{\mathrm{C}_{#1}^{#2}}
\newcommand{\CCmono}{\CC{\psi}{\mono}}
\newcommand{\CCfourier}{\CC{\psi}{\fourier_k}}

\newcommand{\deq}{:=}

\newcommand{\disj}{\textsc{disj}\xspace}
\newcommand{\kdisj}{\textsc{$k$-disj}\xspace}
\newcommand{\lkdisj}{\textsc{($\ell k$)-disj}\xspace}

\newcommand{\ordisj}{\textsc{or-disj}\xspace}
\newcommand{\orkdisj}{\textsc{or-$k$-disj}\xspace}
\newcommand{\orddisj}{\textsc{or-$d$-disj}\xspace}
\newcommand{\disjname}{\textsc{set-disjointness}\xspace}

\newcommand{\sparsedisj}{\textsc{sparse-set-disjointness}\xspace}

\newcommand{\E}{\mathbb{E}}
\newcommand{\Ex}{\E}
\newcommand{\eps}{\epsilon}
\newcommand{\epsh}{\ell}

\newcommand{\fourier}{\textsc{fourier}}
\newcommand{\mono}{\textsc{mono}\xspace}
\newcommand{\monotone}{\textsc{mono}}

\newcommand{\N}{\mathbb{N}}

\newcommand{\ZZ}{\mathbb{Z}}

\newcommand{\hastad}{H{\aa}stad\xspace}

\def\l{\ell}
\def\cC{\mathcal{C}}
\def\cP{\mathcal{P}}

\title{Distance-Sensitive Property Testing Lower Bounds}
\author{Joshua Brody \\ Swarthmore College \\ joshua.e.brody@gmail.com
  \and Pooya Hatami \\ University of Chicago \\ pooya@cs.uchicago.edu
}

\begin{document}
\maketitle

\begin{abstract}
In this paper, we consider several property testing problems and ask
how the query complexity depends on the distance parameter $\eps$.  We
achieve new lower bounds in this setting for the problems of testing
whether a function is monotone and testing whether the function has
low Fourier degree.  For monotonicity testing, our lower bound matches
the recent upper bound of Chakrabarty and Seshadhri~\cite{CS:13}.
\end{abstract}

\newpage
\section{Introduction}
Property Testing is a subfield which seeks to understand what can be
learned about a large object given limited access to the object itself.
In a typical setup, a property tester is a randomized algorithm that,
given a large object as input, must (i) accept with probability $2/3$
if the object has a certain property, and (ii) reject with probability
$2/3$ if the object is \emph{far} from having said property.  We
charge the tester for the number of queries it makes to the object,
and hope that it runs in time significantly sublinear in the size of
the object.  The \emph{query complexity} of a property $P$, denoted
$Q(P)$, is the minimum number of queries needed to test an object for
$P$.

Property testing has been considered for many different properties on
many classes of objects, including testing properties of graphs,
probability distributions, and functions.  See surveys by
Goldreich~\cite{Goldreich10,Goldreich11} and Ron~\cite{Ron:09} for
comprehensive development.
We focus on testing properties of functions.  A property of functions
$f: D \rightarrow R$ is a subset $P \subseteq \{f : D \rightarrow
R\}$.  A query to $f$ is $f(x)$ for some $x \in D$.  We say that $f$
is $\eps$-far from $P$ if $\Pr[f(x) \neq g(x)] \geq \eps$ for all $g
\in P$.

Since the seminal work of Rubinfeld and Sudan~\cite{RS96}, several
properties have been considered, including testing
linearity~\cite{BLR93}, junta
testing~\cite{FKR+:04,Bla:08,Bla:09}, testing whether a function is
isomorphic to a given function~\cite{BO:10,AB:10,CGM11}, and testing
whether a function can be computed by various weak models of
computation, including size-$s$ decision trees~\cite{DLM+:07} and
small-width OBDDs~\cite{Gol:10,RT09,RT10}. 

A variety of techniques have been developed for designing property
testing algorithms thus proving testing upper bounds. However, as is
often the case in theoretical computer science, lower bounds are
harder to come by. Although several lower bounds are known for
different property testing problems, very few general techniques are
known beyond the use of Yao's minimax lemma. Recently, Blais
et. al.~\cite{BBM:12} came up with a new technique to prove property
testing lower bounds, using known lower bounds for randomized
communication complexity problems. In particular, they show how to
reduce certain communication complexity problems to testing problems,
thus showing that communication lower bounds imply lower bounds for
property testing. They show that this technique is indeed powerful by
applying it on several testing problems and improving on some previous
known lower bounds for testing $k$-linearity, $k$-juntas, Fourier
degree$\leq k$, etc. It has not been obvious how to come up with lower
bounds with dependence on the distance parameter $\epsilon$ using this
technique.  In this work, we extend the technique of Blais et al.~and
prove testing lower bounds that depend on $\eps$.

\subsection{Our Results}
\paragraph{Monotinicity Testing:} Our first result is a new lower bound for monotonicity testers.  A
function $f:\pmbits{n} \rightarrow R$ is \emph{monotone} if $f(x) \leq
f(y)$ whenever $x \leq y$.\footnote{We say that $x \leq y$ if $x_i
  \leq y_i$ for all $1\leq i \leq n$.}  Let \mono be the property that
a function is monotone.
Monotonicity testing on various domains has been extensively studied
(\cite{BatuRW99,ErgunKKRV00,GGL+:00,DGL+:99,FLN+:02,BGJ+:09,BCGM10}).
However, for functions on the boolean hypercube, progress remained
elusive until very recently, and gaps between known upper and lower
bounds remain.  For large range sizes, recent progress has closed this
gap considerably: the query complexity lies between $O(n/\eps)$
(Chakrabarty and Seshadhri \cite{CS:13}) and $\Omega(\min\{n,
|R|^2\})$ (Blais et al.~\cite{BBM:12}).  In this work, we give a new
lower bound that completely closes this gap for large $|R|$.
\begin{infthm}[Theorem~\ref{thm:mono}]
  Testing \mono requires $\Omega\left(\min\{n/\eps,
  |R|^2/\eps\}\right)$ queries.
\end{infthm}

Note that for $|R| = \Omega(\sqrt{n})$, Theorem~\ref{thm:mono} is
tight, even for subconstant $\eps$.  Establishing lower bounds
sensitive to the distance parameter is not just a trivial pursuit.
Indeed, recent work suggests that for monotonicity testing of boolean
functions, understanding how the distance parameter affects the query
complexity is key to understanding the overall difficulty of
monotonicity testing.  In another very recent work, Chakrabarty and
Seshadhri~\cite{CS:13b} give a new tester for boolean functions.
Their tester is nonadaptive, has one-sided error, and makes
$\tilde{O}(n^{5/6}\eps^{-5/3})$ queries.  Their result is surprising
and somewhat counterintuitive, because their tester is a path tester,
and an earlier lower bound of~\cite{BCGM10} states that
$\Omega(n/\eps)$ queries are required.  This lower bound
crucially assumes a linear dependence on $\eps$ in the distance
parameter.  In this way, focusing on bounds sensitive to the distance
parameter $\eps$ appears key to obtaining tight bounds.

\paragraph{Testing Fourier Degree:} Our second result is for testing whether a boolean function
$f:\pmbits{n} \rightarrow \pmbit$ has low Fourier degree.  $f$ has
\emph{Fourier degree $k$} if all nonzero Fourier coefficients have
degree at most $k$.  Let $\fourier_k$ denote this property.

Upper bounds of $2^{O(k)}/\eps^2$ are known for testing whether a Boolean
function has Fourier degree $\leq k$ or is $\epsilon$-far from any
Boolean function with Fourier degree $k$~\cite{CGM11a, DLM+:07}. The
best lower bound known on this problem has been
$\Omega(\min\{k,n-k\})$~\cite{BBM:12,CGM11}, which holds for any $\epsilon\leq
1/2$. In this paper we show a lower bound of
$\Omega(\min\{k,n-k\}/\sqrt{\epsilon})$ for testing Fourier degree.
\begin{infthm}[Theorem~\ref{thm:fourier}]
  Testing $\fourier_k$ requires $\Omega(\frac{\min\{k,n-k\}}{\sqrt{\eps}})$
  queries.
\end{infthm}
To our knowledge, this is the first lower bound for $\fourier_k$
sensitive to $\eps$.  While the bound is far from tight (in terms of
both $k$ and $\eps$) even this analysis requires heavier
machinery than our lower bound for \mono testers.

Our final result is a distance-sensitive lower bound for
\emph{approximate Fourier degree} testing.  In this problem, the goal
is to distinguish functions of Fourier degree at most $k$ from those
that are far from any function with Fourier degree at most $n-k$.
\begin{infthm}[Theorem~\ref{thm:main1}]
  Any nonadaptive tester for approximate Fourier degree requires
  $\Omega(1/\eps)$ queries.
\end{infthm}
Chakrabarty et al.~\cite{CGM11a} gave a lower bound of $\Omega(k)$ for
the same problem.  Thus, the combined lower bound is $\Omega(k+1/\eps)$.

\subsection{Outline}
In Section~\ref{section:Method}, we formalize notation and describe
the tools we use.  Section~\ref{sec:mono} develops the lower bound for
monotonicity testers, and Section~\ref{sec:fourierdegree} gives our
lower bounds on Fourier degree.

\section{Preliminaries and Notation}\label{section:Method}
We use $[n]$ to denote the set $\{1,\ldots, n\}$.  The \emph{Hamming
  weight} of a string $x \in \pmbits{n}$, denoted $|x|$, is the number
of $i$ such that $x_i=1$.  We occasionally abuse notation and
associate a string $x \in \pmbits{n}$ with the corresponding set $\{i
: x_i = +1\}$.  When working with multiple strings $x_1,\ldots, x_\ell
\in \pmbits{n}$, we use double subscripts $x_{ij}$ to denote the $j$-th
bit of $x_i$.  This paper analyzes univariate functions
$f:\pmbits{n}\rightarrow \pmbit$, but in our proof of the monotonicity 
testing lower bound it will often be useful to think
of $f$ as a bivariate function.  Hence, we abuse notation somewhat and
use $f(t,z)$ to refer to a univariate function whose input is the
concatenation of $t$ and $z$.

\subsection{{\bf Communication Complexity and Property Testing Lower Bounds}}
In this section we will give a brief introduction to communication
complexity, and state known lower bounds for the famous set
disjointness problem.  The two-party communication model was
introduced by Andrew C. Yao~\cite{Yao79} in 1979.  In this paper,
we are chiefly concerned with (public-coin) randomized communication
complexity.  In a typical communication protocol, there are two
players, Alice and Bob.  Alice receives an input $x$; Bob receives an
input $y$, and they wish to communicate to jointly compute some
function $f(x,y)$ of their inputs.  In a randomized protocol, Alice
and Bob have access to a shared random string $R$ which they can use
to select what messages to send.  Furthermore, players are allowed a
small amount of error.  The \emph{communication cost} of a protocol is
the worst-case maximum number of bits sent, taken over all possible
inputs and values of $R$.  The \emph{communication complexity} of a
function $f$ is the minimal communication cost of a protocol computing $f$.

\begin{definition}
  The \emph{bounded-error randomized communciation complexity} of $f$,
  denoted by $R_\eps(f)$ is the minimum cost of a randomized
  communication protocol that on all inputs $(x,y)$ computes $f$ with
  success probability $1-\eps$.  We set $R(f) \deq R_{1/3}(f)$.
\end{definition}

Of particular interest to us is the \disjname problem.

\begin{definition}[Set Disjointness]
  Alice and Bob are given $x$ and $y$, $x,y\in \pmbits{n}$, and
  need jointly to compute $$\disj_n(x,y)= \vee_{i=1}^{n} (x_i \wedge y_i),$$
  where $(a \wedge b)=1$ if $a=b=1$, and $-1$ otherwise.
\end{definition}
We drop the subscript when it is clear from context.
\disjname is perhaps the primary communication problem used to prove
lower bounds in other areas of computer science.  Obtaining a tight
lower bound on its communication complexity was an important problem
first solved by Kalyanasundaram and
Schnitger~\cite{KalyanasundaramS92}, then simplified by
Razborov~\cite{Razborov90} and later by Bar-Yossef et
al.~\cite{BarYossefJKS02} using information complexity.

\begin{theorem}[\cite{KalyanasundaramS92,Razborov90,BarYossefJKS02}]
  $R(\disj_n) = \Omega(n)$.
\end{theorem}

We will also be interested in the \sparsedisj problem, denoted
$\kdisj$.  In this version, we are promised that the inputs each have
Hamming weight $k$ ($|x| = |y| = k$) and that $|x\cap y| \leq 1$;
i.e., the intersection, if it exists, is unique.  This
``unique-intersection'' promise is implicit in the lower bounds of
Kalyanasundaram and Schnitger and of Razborov.  The following lower
bound from~\cite{BBM:12} is a straightforward reduction from $\disj$.
\begin{lemma}[\cite{BBM:12}]\label{lem:kdisj-lb}
$R(\kdisj_n)=\Omega(\min\{2k,n-2k\})$.
\end{lemma}
This tightens a folklore $\Omega(k)$ lower bound, which holds as long
as $k \leq n/3$.  Note that this is essentially tight, as shown in a
clever protocol by \hastad and Wigderson~\cite{HastadW07}.

\begin{theorem}[\cite{HastadW07}]
  $R(\kdisj_n) = O(k)$.
\end{theorem}

It will be convenient for our applications to consider the following
direct-sum variants of $\disj$.
\begin{definition}
  The functions $\ordisj_m^\ell, \orkdisj_m^{\ell} : \pmbits{\ell m}
  \times \pmbits{\ell m} \rightarrow \pmbit$ are defined
  as $$\ordisj_m^{\ell}(x_1,\ldots, x_\ell,y_1,\ldots, y_\ell) \deq
  \bigvee_{i=1}^{\ell} \disj_m(x_i,y_i)\ .$$ $$\orkdisj_m^{\ell}(x_1,\ldots,
  x_\ell,y_1,\ldots, y_\ell) = \bigvee_{i=1}^{\ell}\kdisj_m(x_i,y_i)\ .$$
\end{definition}

\begin{lemma}\label{lem:ordisj}
  The following direct-sum properties hold:
  \begin{enumerate}
  \item
    $R(\ordisj_m^{\ell}) = \Omega(\ell m)$.  
  \item
    $R(\orkdisj_m^{\ell}) = \Omega(\min\{\ell k, \ell(m-2k)\})$.
  \end{enumerate}
\end{lemma}
\begin{proof}
  We include the proof for $\orkdisj_m^{\ell}$; the proof for
  $\ordisj_m^{\ell}$ is similar.  Let $x_1,\ldots, x_\ell$ and
  $y_1,\ldots, y_\ell$ be inputs to $\orkdisj_m^{\ell}$, and let $x$
  and $y$ be the concatenation of $x_1,\ldots, x_\ell$ and
  $y_1,\ldots, y_\ell$ respectively.  Note that $x$ and $y$ each have
  Hamming weight $\ell k$, and furthermore $x$ and $y$ intersect iff
  $x_i$ and $y_i$ intersect for some $i$.  Thus, we have
  \begin{align*}
    \orkdisj_m^{\ell}(x_1,\ldots, x_\ell,y_1,\ldots, y_\ell) &=
    \bigvee_{i=1}^\ell \kdisj_m(x_i,y_i) \\
    &= \bigvee_{i=1}^\ell \bigvee_{j=1}^m x_{ij}\wedge y_{ij} \\
    &= \lkdisj_{\ell m}(x,y)\ .
  \end{align*}
  Thus, any protocol for $\orkdisj_m^{\ell}$ is also a protocol for
  $\lkdisj_{\ell m}$.  From Lemma~\ref{lem:kdisj-lb} it follows that
  $R(\orkdisj_m^{\ell}) = R(\lkdisj_{\ell m}) = \Omega(\min\{2\ell k,
  \ell m - 2\ell k\})$.
\end{proof}

Next, we summarize the terminology and main lemma for proving testing
lower bounds via communication complexity, reformulating the notation
in a way convenient for our results.  For more details, consult the
work of Blais et al.~\cite{BBM:12}.

\begin{definition}[Combining Operator]
  A \emph{combining operator} $\psi$ takes as input two functions $f,g
  : \pmbits{n}\rightarrow \{-1,+1\}$ and returns a
  function $h: \pmbits{n}\rightarrow R$.

  A combining operator is \emph{simple} if for all $f,g$ and for all
  $x\in \pmbits{n}$, $h(x)$ can be computed given only $x$ and the queries $f(x)$
  and $g(x)$.
\end{definition}

For a property $\calP$ and combining operator $\psi$, let
$\CC{\psi}{\calP}$ denote the communication problem where Alice and
Bob receive $f$ and $g$ respectively and wish to determine if
$\psi(f,g)$ has property $\calP$ or is far from having $\calP$.  The
connection between property testing and this communication game is
captured in the following lemma.

\begin{lemma}[Main Reduction Lemma~(\cite{BBM:12}, Lemma 2.4)]\label{lem:mrl}
  For any simple combining operator $\psi$ and any property $\calP$,
  we have $$R(\CC{\psi}{\calP}) \leq 2Q(\calP)\ .$$
\end{lemma}

\subsection{{\bf Fourier Analysis of Boolean Functions}}
Our result on testing Fourier degree uses Fourier Analysis.  We briefly present
basic definitions and results here.
Consider the $2^n$-dimensional vector space of all functions $f:
\{-1,1\}^n \rightarrow \mathbb{R}$. An inner product on this space can
be defined as follows $$\langle f,g\rangle = \frac{1}{2^n} \sum_{x\in
  \{-1,1\}^n} f(x)g(x) = \Ex_x[f\cdot g],$$ where the latter expectation is
taken uniformly over all $x\in \{-1,1\}^n$. This defines the
$l_2$-norm $$||f||_2 = \sqrt{\langle f,f\rangle} = \sqrt{\Ex_x[f^2]}.$$

\begin{defn}
For $S\subseteq [n]$, the character $\chi_S: \{-1,1\}^n \rightarrow
\{-1,1\}$ is defined as $$\chi_S(x)= \prod_{i\in S} x_i.$$

The set of characters forms an orthonormal basis for the inner product
space. Hence, every function $f:\{-1,1\}^n \rightarrow \mathbb{R}$ can
be written uniquely as $$ f= \sum_S \langle f, \chi_S\rangle\cdot  \chi_S.$$
The above equation is referred to as the Fourier expansion of $f$, and
the Fourier coefficient of $f$ corresponding to set $S$ is defined
as $$\widehat{f}(S)= \langle f,\chi_S\rangle.$$
\end{defn}

{\em Parseval's Identity} states that 
\begin{equation}
\| f\|_2^2 = \sum_{S\subseteq[n]} \hat{f}(S)^2.
\end{equation}

\begin{defn}
The \textit{Fourier degree} of a Boolean function
$f:\{-1,1\}^n\rightarrow \{-1,1\}$ is equal to maximum $k>0$ such that
there exists $S\subseteq [n]$, $|S|=k$, for which $\widehat{f}(S)\neq
0$.
\end{defn}

\section{Testing Monotonicity}\label{sec:mono}
Before getting into the full proof of our lower bound, we give a
high-level description.  We use the technique of Blais et al.~and
reduce from $\ordisj$.  The function $h:\pmbits{n} \rightarrow R$ we
create uses the first $\log(1/\eps)$ bits as an index $t$.  Then, the
value of $h(t,z)$ is essentially the function used by Blais et al. in
their $\Omega(n)$ lower bound for testing monotonicity, evaluated on
the $t$-th pair of inputs for $\ordisj$.  If $x_t$ and $y_t$ intersect
in, say, the $i$-th bit, then $h$ will violate monotonicity on edges
in the $i$-th direction, but \emph{only} when the value of the index
is $t$.  This happens with probability $\eps$ (taken over a random
input to $h$).  Therefore, the overall distance to monotonicity is
$\Omega(\eps)$.

There is one final complication.  We need to ensure that monotonicity
is not spuriously violated in the $i$-th direction when $i$ is one of
the coordinates that defines $t$.  To manage this, we
increase the value of $h$ by $\Omega(|t|)$.  This ensures that
monotonicity is not violated in a direction corresponding to one of
the bits that make up $t$, no matter what happens with the rest of the
input to $h$.

Modulo a few technical complications and adjustment of variables, this
completes the proof that testing monotonicity for functions with high
range size requires $\Omega(n/\eps)$ queries.  For smaller range
sizes, we adopt some range reduction tricks from Blais et
al.~\cite{BBM:12} One technical complication arises, but conceptually,
the reductions are the same.

Before proving the lower bound for \mono, we define some functions and
demonstrate some related basic facts that will be useful for our
proof.  For $i \in \N$ and $x \in \pmbits{*}$, let
$x\iplus$ and $x\iminus$ denote the strings obtained from $x$ by
setting $x_i = +1$ and $x_i = -1$ respectively.

The following fact says that when flipping the $i$th bit of a string,
the value of the character changes only when $i\in S$.
\begin{fact}\label{fact:chi}
  For any $S \subseteq [n]$ and $x \in \pmbits{n}$, $\chi_S(x\iplus) =
  \chi_S(x\iminus)$ if and only if $i \not \in S$.
\end{fact}
In Section~\ref{sec:mono-lb}, we will need to manipulate sums of
certain character functions.  The corollary below easily follows from
Fact~\ref{fact:chi}.
\begin{corollary}\label{cor:chi}
  For any $S,T \subseteq [n]$, $i \in [n]$, and $x \in \pmbits{n}$, we
  have 
  \begin{itemize}
  \item
    $\chi_S(x\iplus) + \chi_T(x\iplus) - \chi_S(x\iminus) -
    \chi_T(x\iminus) = -4$ if $i \in S\cap T$ and $\chi_S(x\iplus) =
    \chi_T(x\iplus) = -1$,
  \item
    $\chi_S(x\iplus) + \chi_T(x\iplus) - \chi_S(x\iminus) -
    \chi_T(x\iminus) \geq -2$ otherwise.
  \end{itemize}
\end{corollary}

Next, we generalize character functions to operate on lists of
strings.  Recall that we associate strings $x_j \in \pmbits{m}$ with
the corresponding set $\{r : x_{jr} = 1\}$.
\begin{definition}
  Let $x = (x_1,\ldots, x_\ell) \in \pmbits{\ell m}$.  The
  function $f_x : \pmbits{\log(\ell) + m}\rightarrow
  \pmbit$ is defined as $$f_x(j, z) \deq \chi_{x_j}(z)\ .$$
\end{definition}
We adopt the convention of mapping $x \rightarrow f_x$ and $y
\rightarrow g_y$.
The following fact is an analogue of Corollary~\ref{cor:chi} and
specifies what happens to the sum of (generalized) character functions
when we flip an input bit from $-1$ to $1$.
\begin{fact}\label{fact:chi2}
  For any $x,y \in \pmbits{\ell m}, t \in \pmbits{\log(\ell)}$, and
  any $z \in \pmbits{m}$, the following hold.
  \begin{enumerate}
  \item
    For all $i \in [\log(\ell)]$, we have $$f_x(t\iplus,z) +
    g_y(t\iplus,z) - f_x(t\iminus,z) - g_y(t\iminus,z) \geq -4\ .$$
  \item
    For all $i \in [m]$, we have $$f_x(t,z\iplus) + g_y(t,z\iplus) -
    f_x(t,z\iminus)- g_x(t,z\iminus) = -4\ ,$$ if $i \in x_t \cap y_t$
    and $f_x(t,z\iplus) = g_y(t,z\iplus) = -1$.  Otherwise, we have
    $$f_x(t,z\iplus) + g_y(t,z\iplus) - f_x(t,z\iplus) - g_y(t,z\iplus)
    \geq -2\ .$$
  \end{enumerate}
\end{fact}

\subsection{The Monotonicity Lower Bound}\label{sec:mono-lb}

\begin{theorem}\label{thm:mono}
  Fix $\eps$ such that $2^{-n/2} < \eps < 1/10$.  Then,
  testing $h:\bits{n} \rightarrow R$ for monotonicity requires
  $\Omega(\min\{n/\eps, |R|^2/\eps\}))$ queries.
\end{theorem}
Define $\ell \deq 1/8\eps$, and assume $\log(\ell)$ is an integer.
Let $m \deq n-\log(\ell)$.
Theorem~\ref{thm:mono} follows directly from the following claims:

\begin{claim}\label{claim:part1}
  If $|R| = \Omega(n)$ then testing $h:\bits{n}\rightarrow R$ for
  monotonicity requires $\Omega(n/\eps)$ queries.
\end{claim}
\begin{claim}\label{claim:part2}
  There exists a constant $c$ such that if $|R| \geq c \sqrt{\nh}$ then
  testing $h:\bits{n} \rightarrow R$ for monotonicity requires
  $\Omega(n/\eps)$ queries.
\end{claim}
\begin{claim}\label{claim:part3}
  If $R = O(\sqrt{\nh})$ then testing $h:\bits{n} \rightarrow R$ for
  monotonicity requires $\Omega(|R|^2/\eps)$ queries.
\end{claim}

In the rest of this section we prove the above claims.
\begin{proof}[Proof of Claim~\ref{claim:part1}]
  We reduce from $\ordisj_\nh^{\epsh}$.
  Let $\psi$ be the combining operator that, given functions $f,g:
  \pmbits{n} \rightarrow \pmbit$, returns the function
  $h:\pmbits{n}\rightarrow \ZZ$ defined by
  $$h(t,z) \deq 4|t| + 2|z| + f(t,z) + g(t,z)\ .$$ Define
  $\CC{\psi}{\mono}$ to be the communication game where Alice and Bob
  are given functions $f$ and $g$ respectively and must test whether
  $h$ is monotone or $\eps$-far from monotone.  By
  Lemma~\ref{lem:mrl} and Lemma~\ref{lem:ordisj}, we have
  $$ 2\,Q(\monotone) \geq R(\CC{\psi}{\monotone}) \quad \mbox{and}
  \quad R(\ordisj_{\nh}^{\epsh}) = \Omega(\epsh\nh) =
  \Omega(n/\eps)\ .$$ We complete the proof by showing that
  $R(\CCmono) \geq R(\ordisj_{\nh}^{\epsh})$.  Given
  $\ordisj_\nh^\epsh$ inputs \mbox{$x = (x_1,\ldots, x_\ell)$} and \mbox{$y =
  (y_1,\ldots, y_\ell)$}, Alice and Bob construct functions $f_x$ and
  $g_y$ respectively.  We claim that when $x_t$ and $y_t$ are disjoint
  for all $t \in [\ell]$, $h$ is monotone, and conversely when $x_t$
  and $y_t$ intersect for some $t$, $h$ is $\eps$-far from monotone.
  To see this, fix an index $t$ and a string $z$.

  Our first task is to show that monotonicity cannot be violated in a
  direction corresponding to the index $t$.  Suppose that $i
  \in [\log(\ell)]$, and consider $h(t\iplus,z) - h(t\iminus,z)$.  We
  have
  \begin{align*}
    h(t\iplus,z) - h(t\iminus,z) &= 4|t\iplus| + 2|z| +
    f_x(t\iplus,z) + g_y(t\iplus,z) \\ &\quad - 4|t\iminus| - 2|z| -
    f_x(t\iminus,z) - g_y(t\iminus,z) \\ &= 4 + f_x(t\iplus,z) +
    g_y(t\iplus,z) - f_x(t\iminus,z) - g_y(t\iminus,z) \\ &\geq 0\ ,
  \end{align*}
  where the inequality holds by Fact~\ref{fact:chi2}.  This shows that
  monotonicity is never violated in the direction of an index
  coordinate.

  Our next task is to show that if $i$ is a cooordinate corresponding
  to $z$, then monotonicity is violated if and only if $x_t$ and $y_t$
  intersect.  Formally, let $i\in [m]$, and consider
  $h(t,z\iplus)-h(t,z\iminus)$.  We have
  \begin{align*}
    h(t,z\iplus) - h(t,z\iminus) &= 4|t| +2|z\iplus| + f_x(t,z\iplus)
    + g_y(t,z\iplus) \\ &\quad -4|t| - 2|z\iminus| - f_x(t,z\iminus)
    - g_y(t,z\iminus) \\ &= 2 + f_x(t,z\iplus) + g_y(t,z\iplus) -
    f_x(t,z\iminus) - g_y(t,z\iminus) \ .
  \end{align*}
  By Fact~\ref{fact:chi2}, this is negative if and only if $i \in x_t
  \cap y_t$ and $f_x(t,z\iplus) = g_y(t,z\iminus) = -1$.  Furthermore,
  this latter condition happens with probability $1/4$, where the
  probability is over a random $z$.

  Together, these cases show that $h$ is monotone when $x_t \cap y_t =
  \emptyset$ for all $t$.  On the other hand, if $x_t$ and $y_t$
  intersect for some $t$, fix $i \in x_t\cap y_t$.  For a random $z$,
  $h(t,z\iplus) - h(t,z\iminus) < 0$ with probability $1/4$; thus, a
  $(1/4)$-fraction of edges in the $i$th direction are violated when
  the index equals $t$, and a random index equals $t$ with probability
  $1/\epsh = 8\eps$.  To attain a monotone function, we must change at
  least one endpoint of each the violating edges in the $i$th
  coordinate.  Therefore, we must change at least $2^n\cdot(1/4)\cdot
  1/\epsh \cdot (1/2) = \eps \cdot 2^n$ points to arrive at a monotone
  function.  Hence, $h$ is $\eps$-far from monotone.
\end{proof}

\begin{proof}[Proof of Claim~\ref{claim:part2}]
  Fix a constant $\ch$ such that $\big| |z| - \nh/2 \big| \leq
  \ch\sqrt{\nh}$ with probability $15/16$.  Using $h$ as defined in
  Claim~\ref{claim:part1}, we define $h':\pmbits{n}\rightarrow R$ by
  $$h'(t, z) = \begin{cases} +\infty & \text{if } |z| \geq \nh/2 + \ch\sqrt{\nh} \ . \\
    -\infty & \text{if } |z| \leq \nh/2 - \ch\sqrt{\nh} \ , \\
    h(t, z) & \text{otherwise.}
  \end{cases}$$

  Note that $|R| = c\sqrt{m}$ for some
  constant $c$ depending on $\ch$.  Recall that a $(1/4)$-fraction of
  edges in the $i$-th direction are violated when $i \in x_t \cap y_t$
  and the index equals $t$, and therefore we need to change at least
  $(1/8)$ of $h(t,z)$ to get a monotone function.  By our choice
  of $\ch$, $h(t,z) \neq h'(t ,z)$ with probability at most
  $1/16$ (over a random $z$ and any fixed $t$).

  It follows that for each $t$ where $x_t$ and $y_t$ intersect, we
  need to change $h'(t,z)$ for at least a $(1/16)$-fraction of $z$ to
  get a monotone function.  Overall, the distance to monotonicity is
  at least $1/(16\ell) = \eps/2$.  Rescaling $\eps$ completes the proof.
\end{proof}

\begin{proof}[Proof of Claim~\ref{claim:part3}]
  We use a claim from~\cite{BBM:12}.

  \begin{claim}[\cite{BBM:12}, Claim 4.2]\label{claim:BBM}
    If there is a $q$-query algorithm that tests
    $f:\bits{n}\rightarrow R$ for monotonicity when $|R| =
    O(\sqrt{n})$ then there is a $q$-query algorithm that tests
    $g:\bits{m}\rightarrow R$ for monotonicity where $m =
    \Omega(|R|^2)$.
  \end{claim}

  Claim~\ref{claim:part2} shows that testing $g$ requires
  $\Omega(m/\eps)$ queries.  Thus, testing $f$ requires
  $\Omega(m/\eps) = \Omega(|R|^2/\eps)$ queries.
\end{proof}

\section{Testing Fourier Degree}\label{sec:fourierdegree}

In this section we present our lower bound for testing
whether a given function has low Fourier degree. For
convenience, in the context of Fourier analysis we consider the
Boolean function to be of the form $f: \{-1,1\}^n\rightarrow
\{-1,1\}$.

Diakonikolas \textit{et al.}~\cite{DLM+:07} considered the problem of
testing whether a Boolean function $f$ has Fourier degree at most
$k$. They proved a general lower bound of $\tilde{\Omega}(\log k)$,
and a lower bound of $\tilde{\Omega}(\sqrt{k})$ for the non-adaptive
tester with any $\epsilon\leq 1/2$. They also present an algorithm
with $\tilde{O}(2^{6k}/\epsilon^2)$ query complexity for this
problem. Chakraborty \textit{et al.}~\cite{CGM11} and later Blais
\textit{et al.}~\cite{BBM:12} improved the lower bound to
$\Omega(\min\{k, n-k\})$, for any $0\leq \epsilon\leq 1/2$. In this
section we show how to use the communication complexity technique to
prove a lower bound of $\Omega(\min\{k,n-k\}/\sqrt{\epsilon})$ on
testing whether a Boolean function is of Fourier degree at most $k$.

\begin{theorem}\label{thm:fourier}
  Let $\epsilon \geq 2^{-k-1}$. Testing whether a Boolean function
  $f:\{-1,1\}^n\rightarrow \{-1,1\}$ has Fourier degree $\leq k$ or
  $\epsilon$-far from any Boolean function with Fourier degree $\leq
  k+1$, requires $\Omega(\frac{\min\{k,n-k\}}{\sqrt{\epsilon}})$
  queries.
\end{theorem}

\begin{remark}
  Note that the case when $\epsilon \leq 2^{-k-1}$ is at least as hard
  as when $\epsilon=2^{-k-1}$ and thus an exponential lower bound of
  $\Omega(2^k)$ for such $\epsilon$ follows immediately from the above
  theorem.
\end{remark}

In order to prove our lower bound we first need to construct Boolean
functions with certain properties.

\subsection{{\bf Our Constructions of Functions}}\label{sec:fourier}
In this section we give a method how to construct functions which are
of Fourier degree $\leq k$ and functions that are far from having
Fourier degree at most $k$.

\begin{definition}[Functions Defined by Index Selectors]\label{def:fun-idx-sel}
Let $\l>0$ be an integer. Call a map
$$ \mathcal{C}^\l: \{-1,1\}^\l \rightarrow \cP([n]\backslash
\{1,\ldots, \l\})
$$ to be an {\em index selector}, where $\cP(\cdot)$ denotes the power
set. Given an index selector $\mathcal{C}^\l$ one can define a Boolean
function $f^\l:\{-1,1\}^n\rightarrow \{-1,1\}$ as following
$$
f^\l(x_1,...,x_n) = \chi_{\cC(x_1,\ldots,x_\l)}(x_1,...,x_n),
$$
where $\chi_A(x_1,...,x_n):= \prod_{i\in A} x_i$ for $A\subseteq [n]$.
We call $f^\l$ an \emph{index selector function}.
\end{definition}

In the next three propositions we show how the cardinalities of the
sets $\cC^\l(a_1,...,a_\l)$ can lead $f^\l$ to be of low Fourier
degree, or to be far from any Boolean function with low Fourier
degree. For the sake of simplicity we will often avoid the superscript
$\l$ in $\cC^\l(a_1,\ldots, a_\l)$ and simply write
$\cC(a_1,\ldots,a_\l)$.

\begin{prop}\label{prop:degk}
  The Boolean function $f^{\l}:\{-1,1\}^n\rightarrow \{-1,1\}$ as
  described above, is of Fourier degree $m+\l$ if
  $$
  \forall (a_1,...,a_{\l})\in \{-1,1\}^\l, 
  |\cC^{\l}(a_1,...,a_{\l})|\leq m.
  $$
\end{prop}
\begin{proof}
  We have to prove that $\langle f, \chi_S \rangle = 0$ for any
  $S\subseteq [n]$ with $|S|\geq m+\l+1$.
\begin{align*}
\widehat{f^\l}(S)&= \langle f^\l,\chi_S \rangle= 2^{-n} \sum_{x\in \{-1,1\}^n} \chi_{\cC(x_{[\l]})}(x) \cdot \chi_S(x) \\ &= 2^{-n} \sum_{x_1,...,x_{\l}} \sum_{x_{\l+1}, ..., x_n} \chi_{\cC(x_1,...,x_{\l})}(x_1,...,x_n) \cdot \chi_S(x_1,...,x_n)=0.
\end{align*}
The last equality follows from the fact that $$\sum_{x_{\l+1}, ...,
  x_n \in \{-1,1\}} \chi_{\cC(x_1,...,x_{\l})}(x_1,...,x_n)\cdot
\chi_S(x_1,...,x_n)=0,$$ since $$\exists i \in S :i\notin
\cC(x_1,...,x_{\l}) \cup \{1,...,\l\},$$ because $|S| \geq m+\l+1$.
\end{proof}

\begin{prop}\label{prop:fardeg2k}
The Boolean function $f^\l$ is $1/{2^{\l+1}}$-far from any Boolean
function of Fourier degree $\leq m-1$ if for only one
$(b_1,...,b_\l)\in \{-1,1\}^\l$, $|\cC(b_1,...,b_\l)|\geq m$, and
$$\forall (a_1,...,a_\l)\neq (b_1,...,b_\l): \; |\cC(a_1,...,a_\l)|\leq m-1.$$
\end{prop}
\begin{proof}
 We first prove that for any $U\subseteq \{1,...,\l\}$, the Fourier
 coefficient of $|f^\l|$ at $S:=U\cup \cC(b_1,...,b_\l)$ is equal to
 $1/2^\l$.

\begin{align*}
\widehat{f}^\l(S) &= \langle f^\l, \chi_{U\cup \cC(b_1,...,b_\l}) \rangle = 
2^{-n}\sum_{x\in \{-1,1\}^n} \chi_{x_{[\l]}}(x)\cdot \chi_{U\cup \cC(b_1,...,b_\l)}(x)\\&=
2^{-n} \sum_{x_{[\l]} \in \{-1,1\}^\l} \left(\prod_{i\in U} x_i\right ) \sum_{x_{l+1},...,x_{n}} \chi_{x_{[\l]}}(x)\chi_{\cC(b_1,...,b_\l)}(x)\\&=
2^{-\l} \prod_{i\in U} b_i + 2^{-n}\sum_{x_{[\l]}\neq (b_1,...,b_n)} \left(\prod_{i\in U} x_i\right ) \sum_{x_{\l+1},...,x_{n}} \chi_{x_{[\l]}}(x)\cdot \chi_{\cC(b_1,...,b_\l)}(x) \\&= 
2^{-l} \prod_{i \in U} b_i.
\end{align*}
The last equality follows from the fact that if $(a_1,...,a_\l)\neq
(b_1,...,b_\l)$ then $|\cC(b_1,...,b_\l)|> |\cC(a_1,...,a_\l)|$,
and $$ \sum_{x_{\l+1},...,x_{n}} \chi_{x_{[\l]}}(x)\cdot
\chi_{\cC(b_1,...,b_\l)}(x)=0.
$$

Let $g:\{-1,1\}^n\rightarrow \{-1,1\}$ be a Boolean function with
Fourier degree at most $m-1$, namely we can write
$$
g(x) = \sum_{S\subseteq [n]}^{|S|\leq m-1} \widehat{g}(S) \chi_S(x).
$$ Notice that the distance between the two functions $f$ and $g$ with
range $\{-1,1\}$ can be formulated as $\frac{1}{2}||f-g||^2_2 =
\frac{1}{2}\Ex[(f-g)^2]$. Finally Parseval Identity implies that
\begin{align*}
||f-g||^2_2 &= 
\sum_{S\subseteq [n]}^{|S|\leq m-1} (\widehat{f}(S)-\widehat{g}(S))^2 + \sum_{S\subseteq [n]}^{|S|\geq m} \widehat{f}(S)^2 \\ &\geq 
\sum_{S\subseteq [n]}^{|S|\geq m} \widehat{f}(S)^2 \geq
\sum_{U\subseteq [\l]} \left(2^{-\l} \prod_{i\in U} b_i \right)^2 = 2^{-\l}.
\end{align*}
\end{proof}

\begin{prop}\label{prop:fardegk}
The Boolean function $f^\l$ is $1/{2^{2\l+1}}$-far from any Boolean
function of Fourier degree $\leq m+\l-1$ if for only one
$(b_1,...,b_\l)\in \{-1,1\}^\l$, $|\cC(b_1,...,b_\l)|\geq m$, and
$$\forall (a_1,...,a_\l)\neq (b_1,...,b_\l): \; |\cC(a_1,...,a_\l)|\leq m-1.$$
\end{prop}
\begin{proof}
The proof is similar to the proof of Proposition~\ref{prop:fardeg2k},
with the difference that we only use the fact that
$\widehat{f^\l}([\l]\cup \cC(b_1,...,b_\l)) = 2^{-\l}$, and thus
$f^\l$ is $2^{-2\l-1}$ far from any function with fourier degree
$m+\l-1$.
\end{proof}
\subsection{{\bf Proof of Theorem~\ref{thm:fourier}}}
In this section, we show how to use the communication complexity
technique to prove a lower bound of $\Omega(\frac{\min\{k,
  n-k\}}{\sqrt{\epsilon}})$ on testing whether a Boolean function is
of Fourier degree at most $k$.
\begin{proof}[Proof of Theorem~\ref{thm:fourier}]
  Let $\l$ be the largest integer such that $\epsilon< 2^{-2\l-1}$.
  Notice that $\epsilon \geq 2^{-k-1}$ implies $\l\leq \frac{k}{2}$.
  Also, let $m \deq n - \l$.  Assume that $n-k$ is even, and let $d
  \deq (n-k)/2$.
  We prove that \mbox{$\Omega(2^{\l}\cdot\min\{k,m-k\})$} queries are
  required to test whether a Boolean function has Fourier
  degree $\leq k$ or is $\epsilon$-far from any Boolean function with
  degree $\leq k+1$ by reducing from $\orddisj_m^{2^\ell}$.

  Let $\psi$ be the combining operator that, given functions $f,g :
  \pmbits{n} \rightarrow \pmbit$, returns the function $h$ defiend as
  $h \deq f \cdot g \cdot \chi_{[n]\setminus [\ell]}$. Define
  $\CCfourier$ to be the communication game where Alice and Bob are
  given functions $f$ and $g$ respectively and must test whether $h$
  has Fourier degree at most $k$ or is $\eps$-far from all functions
  of Fourier degree at most $k+1$.
  By Lemma~\ref{lem:mrl} and Lemma~\ref{lem:ordisj}, we have $$
  2\,Q(\fourier_k) \geq R(\CCfourier) \quad \mbox{and} \quad
  R(\orddisj_{m}^{2^\ell}) = \Omega(2^\ell\min\{d, m-d\}) =
  \Omega\left(\frac{\min\{k,n-k\}}{\sqrt{\eps}}\right)\ .$$ We complete the proof by
  showing that $R(\CCfourier) \geq R(\orddisj_{m}^{2^\ell})$.

  Given inputs $x = (x_1,\ldots, x_{2^\l})$ and $y = (y_1,\ldots,
  y_{2^\l})$ to $\orddisj_m^{2^\l}$, Alice and Bob create functions
  $f,g$ by using \emph{index selectors}.  Specifically, Alice defines
  the index selector $C_x$ in the natural way---for any $t \in
  \pmbits{\l}$, set $C_x(t) \deq x_t$.  Bob similarly builds an index
  selector $C_y$ from $y$.  Let $f_x$ and $g_y$ be the functions
  defined by $C_x$ and $C_y$, as described in
  Definition~\ref{def:fun-idx-sel}.  We claim that the combined
  function $h$ is also an index selector function.  To see this, fix
  $a \in \pmbits{\l}$, and notice that $h(t,z) = \chi_{x_t}(z)\cdot
  \chi_{y_t}(z) \cdot \chi_{[m]}(z) = \chi_T(z)$, where
  $T = [m]\setminus(x_t \mathbin{\Delta} y_t)$.  Thus, $h$ describes
  an index selector function for the index selector $D_{xy}$ defined
  by $D_{xy}(t) \deq [m]\setminus(x_t \mathbin{\Delta} y_t)$.

  The index selector $D_{xy}$ is highly structured.  In
  particular, 
  $$|D_{xy}(t)| = 
  \begin{cases}
    k-\l & \text{if } x_t \cap y_t = \emptyset \\ k+2-\l & \text{if }
    |x_t \cap y_t| = 1
  \end{cases}$$

  By Proposition~\ref{prop:degk} and Proposition~\ref{prop:fardegk},
  we have that $h$ has Fourier degree $k$ if $x_t \cap y_t =
  \emptyset$ for all $t$, and that $h$ is $\eps$-far from any Boolean
  function of degree $k+1$ when $x_t$,$y_t$ intersect for a unique
  $t$.  Thus, an answer to $\CCfourier$ gives an answer to
  $\orddisj_m^{2^\ell}$.
\end{proof}

\subsection{{\bf Approximate Fourier degree testing} }
Chakraborty \textit{et al.}~\cite{CGM11a} proved that testing whether a Boolean function has Fourier degree at most $k$ or it is far from any Boolean function with Fourier degree $n-\Theta(1)$ requires $\Omega(k)$ queries. Here we prove an $\Omega(1/\epsilon)$ lower bound for the non-adaptive tester, using Yao's minimax principle. For this we introduce two distributions $D_+$ and $D_-$ on Boolean functions where $D_+$ is a distribution supported only on Boolean functions with Fourier degree $\leq k$ and $D_-$ is only supported on Boolean functions that are $\epsilon$-far from any Boolean function with Fourier degree $\leq n-2k$. This combined with Chakraborty~\textit{et al.}'s result gives an $\Omega(k+\frac{1}{\epsilon})$ lower bound for non-adaptively approximate testing the Fourier degree.

\begin{thm}\label{thm:main1}
Let $\epsilon\geq 2^{-(k/2-1)}$. Non-adaptively Testing whether a Boolean function $f:\{-1,1\}^n\rightarrow \{-1,1\}$ has Fourier degree $\leq k$ or it is $\epsilon$-far from any Boolean function with Fourier degree $\leq n-k$ requires $\Omega(\frac{1}{\epsilon})$ queries. 
\end{thm}

\begin{proof}
Let $\l$ be the largest integer such that $\epsilon< 2^{-\l-1}$. We prove that $\Omega(2^\l)$ queries are required in order to non-adaptively test whether a Boolean function has Fourier degree $\leq k$ or is $\epsilon$-far from any Boolean function with degree $\leq n-k$. Notice that since $\epsilon \geq 1/2^{\frac{k}{2}-1}$ thus $\l\leq \frac{k}{2}-1$.

Let $D_+$ be the distribution obtained by the following random process. Let $\cC^\l(a_1,\ldots,a_\l)$, for every $(a_1,\ldots,a_\l)\in \{-1,+1\}^\l$, be a uniformly chosen subset of size $k/2$ of $\{\l+1,...,n\}$. Let $f^\l$ defined by $\cC^\l$ as described in Section~\ref{sec:fourier}. Proposition~\ref{prop:degk} immediately implies that $f^\l$ has Fourier degree $\leq k$. 

Let $D_-$ be the distribution obtained by the following process. We choose $(b_1,...,b_\l)\in \{-1,1\}^\l$ uniformly at random and choose $\cC^\l(b_1,...,b_\l)$ to be a previously fixed subset of cardinality $n-k+1$ of $\{\l+1,...,n\}$. Also for any $(a_1,...,a_\l)\in \{-1,1\}^l$, where $(a_1,...,a_\l)\neq (b_1,...,b_\l)$, we choose $\cC^\l(a_1,...,a_\l)$ uniformly to be a subset of cardinality $k/2$ of $\{\l+1,...,n\}$. Finally, construct $f^\l$ according to $\cC^\l$. Proposition~\ref{prop:fardeg2k} immediately implies that $f^\l$ is $2^{-\l-1}$-far from any Boolean function with Fourier degree $\leq n-k$. 

Let our final distribution $D$ be that with probability $1/2$ we draw $f^\l$ from $D_+$ and with probability $1/2$ we draw $f^\l$ from $D_-$. Now by Yao's minimax principle if we prove that any deterministic algorithm with less than $2^l/6$ queries makes a mistake with constant probability, this implies that the original testing problem with constant probability of error requires $\frac{2^{\l}}{6}= \Omega(\frac{1}{\epsilon})$ queries. 

For any deterministic set of queries to the function on $d$ inputs $x^1,...,x^d$, with $d \leq \frac{2^{\l}}{6}$,
$$
|\{(a_1,...,a_\l)| (\exists 1\leq i\leq d)  x^i_{[\l]}= (a_1,...,a_\l)\}|\leq d \leq \frac{2^{\l}}{6}.
$$
Therefore the measure of the set of functions from support of $D_-$ for which the deterministic tester has not yet queried any input from the high degree subcube is at least $$
\frac{1}{2} \cdot \frac{2^\l - 2^\l/6}{2^l} = \frac{5}{12} \geq \frac{1}{3}.
$$ 
Since the same is true for $D_+$, the deterministic tester will make an error with probability at least $\frac{1}{3}$.  
\end{proof}



{\small

\begin{thebibliography}{10}

\bibitem{AB:10}
Noga Alon and Eric Blais.
\newblock Testing boolean function isomorphism.
\newblock In {\em Proc. 14th International Workshop on Randomization and
  Approximation Techniques in Computer Science}, pages 394--405, 2010.

\bibitem{BarYossefJKS02}
Ziv {Bar-Yossef}, T.~S. Jayram, Ravi Kumar, and D.~Sivakumar.
\newblock An information statistics approach to data stream and communication
  complexity.
\newblock In {\em Proc. 43rd Annual IEEE Symposium on Foundations of Computer
  Science}, pages 209--218, 2002.

\bibitem{BatuRW99}
Tugkan Batu, Ronitt Rubinfeld, and Patrick White.
\newblock Fast approximate {PCP}s for multidimensional bin-packing problems.
\newblock In {\em Proc. 3rd International Workshop on Randomization and
  Approximation Techniques in Computer Science}, pages 245--256, 1999.

\bibitem{BGJ+:09}
Arnab Bhattacharyya, Elena Grigorescu, Kyomin Jung, Sofya Raskhodnikova, and
  David~P. Woodruff.
\newblock Transitive-closure spanners of the hypercube and the hypergrid.
\newblock Technical Report TR09-046, ECCC, 2009.

\bibitem{Bla:08}
Eric Blais.
\newblock Improved bounds for testing juntas.
\newblock In {\em Proc. 12th International Workshop on Randomization and
  Approximation Techniques in Computer Science}, pages 317--330, 2008.

\bibitem{Bla:09}
Eric Blais.
\newblock Testing juntas nearly optimally.
\newblock In {\em Proc. 41st Annual ACM Symposium on the Theory of Computing},
  pages 151--158, 2009.

\bibitem{BBM:12}
Eric Blais, Joshua Brody, and Kevin Matulef.
\newblock Property testing lower bounds via communication complexity.
\newblock {\em Computational Complexity}, 2012.

\bibitem{BO:10}
Eric Blais and Ryan O'Donnell.
\newblock Lower bounds for testing function isomorphism.
\newblock In {\em Proc. 25th Annual IEEE Conference on Computational
  Complexity}, pages 235--246, 2010.

\bibitem{BLR93}
Manuel Blum, Michael Luby, and Ronitt Rubinfeld.
\newblock Self-testing/correcting with applications to numerical problems.
\newblock {\em J. Comput. Syst. Sci.}, 47:549--595, 1993.
\newblock Earlier version in STOC'90.

\bibitem{BCGM10}
Jop Bri\"et, Sourav Chakraborty, David Garc\'ia-Soriano, and Arie Matsliah.
\newblock Monotonicity testing and shortest-path routing on the cube.
\newblock In {\em Proc. 14th International Workshop on Randomization and
  Approximation Techniques in Computer Science}, 2010.

\bibitem{CS:13b}
Deeparnab Chakrabarty and C.~Seshadhri.
\newblock An o(n) monotonicity tester for boolean functions over the hypercube.
\newblock In {\em Proc. 45th Annual ACM Symposium on the Theory of Computing},
  2013.

\bibitem{CS:13}
Deeparnab Chakrabarty and C.~Seshadhri.
\newblock Optimal bounds for monotonicity and lipschitz testing over hypercubes
  and hypergrids.
\newblock In {\em Proc. 45th Annual ACM Symposium on the Theory of Computing},
  2013.

\bibitem{CGM11a}
Sourav Chakraborty, David Garc\'{\i}a-Soriano, and Arie Matsliah.
\newblock Efficient sample extractors for juntas with applications.
\newblock In {\em ICALP (1)}, pages 545--556, 2011.

\bibitem{CGM11}
Sourav Chakraborty, David Garc\'ia-Soriano, and Arie Matsliah.
\newblock Nearly tight bounds for testing function isomorphism.
\newblock In {\em Proc. 22nd Annual ACM-SIAM Symposium on Discrete Algorithms},
  2011.

\bibitem{DLM+:07}
Ilias Diakonikolas, Homin Lee, Kevin Matulef, Krzysztof Onak, Ronitt Rubinfeld,
  Rocco Servedio, and Andrew Wan.
\newblock Testing for concise representations.
\newblock In {\em Proc. 48th Annual IEEE Symposium on Foundations of Computer
  Science}, pages 549--558, 2007.

\bibitem{DGL+:99}
Yevgeniy Dodis, Oded Goldreich, Eric Lehman, Sofya Raskhodnikova, Dana Ron, and
  Alex Samorodnitsky.
\newblock Improved testing algorithms for monotonicity.
\newblock In {\em Proc. 3rd International Workshop on Randomization and
  Approximation Techniques in Computer Science}, pages 97--108, 1999.

\bibitem{ErgunKKRV00}
Funda Ergun, Sampath Kannan, Ravi Kumar, Ronitt Rubenfeld, and Mahesh
  Viswanathan.
\newblock Spot-checkers.
\newblock {\em J. Comput. Syst. Sci.}, 60:717--751, 2000.

\bibitem{FKR+:04}
Eldar Fischer, Guy Kindler, Dana Ron, Shmuel Safra, and Alex Samorodnitsky.
\newblock Testing juntas.
\newblock {\em J. Comput. Syst. Sci.}, 68:753--787, 2004.

\bibitem{FLN+:02}
Eldar Fischer, Eric Lehman, Ilan Newman, Sofya Raskhodnikova, Ronitt Rubinfeld,
  and Alex Samorodnitsky.
\newblock Monotonicity testing over general poset domains.
\newblock In {\em Proc. 34th Annual ACM Symposium on the Theory of Computing},
  pages 474--483, 2002.

\bibitem{Gol:10}
Oded Goldreich.
\newblock On testing computability by small width {OBDD}s.
\newblock In {\em Proc. 14th International Workshop on Randomization and
  Approximation Techniques in Computer Science}, pages 574--587, 2010.

\bibitem{Goldreich11}
Oded Goldreich.
\newblock A brief introduction to property testing.
\newblock In {\em Studies in Complexity and Cryptography}, pages 465--469.
  2011.

\bibitem{Goldreich10}
Oded Goldreich.
\newblock Introduction to testing graph properties.
\newblock In {\em Studies in Complexity and Cryptography}, pages 470--506.
  2011.

\bibitem{GGL+:00}
Oded Goldreich, Shafi Goldwasser, Eric Lehman, Dana Ron, and Alex
  Samorodnitsky.
\newblock Testing monotonicity.
\newblock {\em Combinatorica}, 20(3):301--337, 2000.

\bibitem{HastadW07}
Johan H{\aa}stad and Avi Wigderson.
\newblock The randomized communication complexity of set disjointness.
\newblock {\em Theory of Computing}, pages 211--219, 2007.

\bibitem{KalyanasundaramS92}
Bala Kalyanasundaram and Georg Schnitger.
\newblock The probabilistic communication complexity of set intersection.
\newblock {\em SIAM J. Disc. Math.}, 5(4):547--557, 1992.

\bibitem{Razborov90}
Alexander Razborov.
\newblock On the distributional complexity of disjointness.
\newblock In {\em Proc. 17thInternational Colloquium on Automata, Languages and
  Programming}, pages 249--253, 1990.

\bibitem{Ron:09}
Dana Ron.
\newblock Algorithmic and analysis techniques in property testing.
\newblock {\em Foundations and Trends in Theoretical Computer Science},
  5(2):73--205, 2009.

\bibitem{RT09}
Dana Ron and Gilad Tsur.
\newblock Testing computability by width two obdds.
\newblock In {\em 13th International Workshop on Randomization and Computation
  (RANDOM)}, 2009.

\bibitem{RT10}
Dana Ron and Gilad Tsur.
\newblock Testing computability by width-2 obdds where the variable order is
  unknown.
\newblock In {\em 7th International Conference on Algorithms and Complexity},
  2010.

\bibitem{RS96}
Ronitt Rubinfeld and Madhu Sudan.
\newblock Robust characterizations of polynomials with applications to program
  testing.
\newblock {\em SIAM J. Comput.}, 25:252--271, 1996.

\bibitem{Yao79}
Andrew~C. Yao.
\newblock Some complexity questions related to distributive computing.
\newblock In {\em Proc. 11thAnnual ACM Symposium on the Theory of Computing},
  pages 209--213, 1979.

\end{thebibliography}
}
\end{document}